\documentclass[a4paper,reqno]{amsart}
\usepackage{amssymb}
\usepackage{latexsym}
\usepackage{amsmath}
\usepackage{euscript}
\usepackage{graphics,color}
\usepackage[all]{xy}
\usepackage[margin=3cm]{geometry}
\usepackage{MnSymbol} 

\newcommand{\be}{\begin{equation}}
\newcommand{\ee}{\end{equation}}
\newcommand{\ba}{\begin{eqnarray}}
\newcommand{\ea}{\end{eqnarray}}
\newcommand{\baa}{\begin{eqnarray*}}
\newcommand{\eaa}{\end{eqnarray*}}
\newcommand{\bb}{\begin{thebibliography}}
\newcommand{\eb}{\end{thebibliography}}

\newcommand{\bi}[1]{\bibitem{#1}}
\newcommand{\lab}[1]{\label{#1}}
\newcommand{\re}[1]{(\ref{#1})}

\newcounter{my}
\newcommand{\he}%
   {\stepcounter{equation}\setcounter{my}%
   {\value{equation}}\setcounter{equation}0%
   }%
\newcommand{\she}%
   {\setcounter{equation}{\value{my}}%
    }%

\renewcommand\t{\tilde}

\newtheorem{pr}{Proposition}

\theoremstyle{definition}

\numberwithin{equation}{section}

\begin{document}

\title[Algebraic discrete Fourier transform]{An algebraic interpretation of the intertwining 
operators associated with the discrete Fourier transform}

\author{Mesuma Atakishiyeva}
\author{Natig Atakishiyev}
\author{Alexei Zhedanov}

\address{Universidad Aut{\'o}noma del Estado de Morelos, Centro de Investigaci{\'o}n en Ciencias, \newline
Cuernavaca, 62250, Morelos, M{\'e}xico}

\address{Universidad Nacional Aut{\'o}noma de M{\'e}xico, Instituto de Matem{\'a}ticas,  Unidad Cuernavaca,  \newline
Cuernavaca, 62210, Morelos, M{\'e}xico}

\address{School of Mathematics, Renmin University of China, Beijing 100872, China}

\begin{abstract}
We show that intertwining operators for the discrete Fourier transform form a cubic 
algebra $\mathcal{C}_q$ with $q$ a root of unity. This algebra is intimately related 
to the two other well-known realizations of the cubic algebra: the Askey-Wilson algebra 
and the Askey-Wilson-Heun algebra.
\end{abstract}

\maketitle

\section{Introduction}
We begin by recalling first that the discrete (finite) Fourier transform (DFT) based 
on $N$ points is represented by an $N\times N$ unitary symmetric matrix $\Phi$  
with entries (see, for example, \cite{McCPar}-\cite{Ata})
\be
\Phi_{kl} = N^{-1/2} \: q^{kl}, \qquad k,l=0,1,\dots, N-1,
\label{DFT}
\ee
where $ q=\exp{\left(2 \pi {\rm i}\, / N\right)} $ is a primitive $N$-th root of unity. 
Note that the matrix $\Phi$ was introduced by Sylvester \cite{Sylv} back in 1867 
and is frequently referred to as Schur's matrix.

In the present work we discuss some additional findings concerning algebraic 
properties of two finite-dimensional intertwining operators, associated with 
the DFT matrix \re{DFT}. These operators are represented by matrices $A$ 
and $B$ of the same size $N\times N$ such that the intertwining relations 
\be
A \,\Phi ={\rm i}\,\Phi A, \qquad B \,\Phi = -\, {\rm i}\, \Phi \,B, 
\lab{inter} \ee
are valid.  The matrices $A$ and $B$ have emerged in a paper \cite{MesNat} 
devoted to the problem of finding an explicit form for the  difference operator 
that governs the eigenvectors of the DFT matrix $\Phi$.

The purpose of this work is to provide a detailed account of a cubic algebra 
at roots of unity $\mathcal{C}_q$, which these intertwining operators form. 

The lay out of the paper is as follows. Section 2 collects all those known facts 
about the  intertwining operators  $A$ and $B$, which are needed in section 3
for deriving an explicit form of an algebra $\mathcal{C}_q$, formed by these 
operators. In section 4  we show how this algebra $\mathcal{C}_q$ is related 
to the another well-known realization of the cubic algebra -- the Askey-Wilson
algebra. Section 5 closes the paper with a brief discussion of the interrelation
between the algebra $\mathcal{C}_q$ and yet another realization of the cubic
algebra -- the Askey-Wilson-Heun algebra. 

\section{Intertwining operators}

This section begins by rederiving an explicit form of the intertwining operators $A$ 
and $B$ \cite{MesNat}. Let us assume that they are of the most general {\it cyclic} 
3-diagonal form, that is,
\[
A=\left[\begin{array}{ccccc}
b_0 & c_0 & 0 & \cdots & a_{N-1} \\
a_0 & b_1 & c_1 &  \cdots & \cdots  \\
0 & \ddots & \ddots & \ddots & 0 \\
\cdots & \cdots & a_{N-3} & b_{N-2} & c_{N-2} \\
c_{N-1} & \cdots & 0 & a_{N-2} & b_{N-1}
\end{array}\right]
\] 
and
\[
B =\left[\begin{array}{ccccc}
\t b_0 & \t c_0 & 0 & \cdots & \t a_{N-1} \\
\t a_0 & \t b_1 & \t c_1 & \cdots  & \cdots  \\
0 & \ddots & \ddots & \ddots & 0 \\
\cdots &\cdots  & \t a_{N-3} & \t b_{N-2} & \t c_{N-2} \\
\t c_{N-1} & \cdots & 0 & \t a_{N-2} & \t b_{N-1}
\end{array}\right],
\]
where $a_k,b_k,c_k, \t a_k,  \t b_k,  \t c_k$, $ k=0,1,\dots, N-1,$  are complex 
parameters. Note that no additional relations between $A$ and $B$ are assumed 
at the outset, however we specify $B$ later on as a Hermitian conjugate matrix 
with respect to $A$, i.e., $ B=A^{\dagger}$.

From the first identity in \re{inter} it follows at once that the elements  $a_k,b_k,c_k$  
of the matrix $A$ are interconnected by the equations
\be 
a_{j-1} q^{-\,k} + b_j + c_j\, q^k = i\left(q^{-j} c_{k-1} + b_k + a_k\,q^j  \right)
\lab{eq_kj}, \ee
so one readily concludes that $a_k,b_k,c_k$ are linear combinations of $q^k, q^{-\,k}$ 
and a constant. A more detailed analysis of equations \re{eq_kj} leads to the {\it general 
solution}
\ba
a_k = - {\rm i} \alpha + \beta \left({\rm i} q^{-\,k} - q^{k+1} \right), \quad b_k 
= \alpha \left( q^k- q^{-\,k} \right), \quad c_k ={\rm i} \alpha + \beta \left(q^{-\,k} 
- {\rm i} q^{k+1} \right), \lab{gen_abc} 
\ea
where $\alpha$ and  $\beta$ are two arbitrary complex parameters.

Quite similarly, one can find that the general solution for the operator $B$ is
\ba
\t a_k = {\rm i} \t \alpha - \t \beta \left( {\rm i} q^{-\,k} + q^{k+1} \right), \quad 
\t b_k = \t \alpha \left( q^k - q^{-\,k} \right), \quad \t c_k = - {\rm i} \t \alpha + 
\t \beta \left(q^{-\,k} + {\rm i} q^{k+1} \right), \lab{gen_tabc} \ea
where $\t \alpha$ and $\t \beta$ is another pair of arbitrary complex parameters. 
Assume now that $B$ is the Hermitian conjugate of $A$. Then we have the conditions
\[
\t a_n = c_n^*, \qquad \t b_n = b_n^*, \qquad \t c_n = a_n^*,
\]
where $a^*$ means complex conjugate of $a$. From this conditions one finds that
\be
\t \alpha = - \,\alpha^*, \qquad q\, \t \beta = -\, \beta^* .
\ee
If one introduces two linear combinations of the operators $A$ and $A^{\dagger}$
of the form 
\[
X=\frac{1}{2}\left(A+A^{\dagger} \right)  \quad \mbox{and} 
\quad  Y=\frac{1}{2{\rm i}}\left(A-A^{\dagger} \right), \lab{XY_def}
\]
then the operators $X$ and $Y$ are Hermitian and play the role of finite-dimensional 
analogs of the operators of the coordinate and momentum in quantum mechanics.

It is natural to require that the coordinate operator should be {\it diagonal}, in complete
agreement with the ordinary meaning of the coordinate operator. From the above relations 
between $a_n,b_n,c_n$ one concludes that the operator $X$ is diagonal if and only if 
$\beta=0$ and $\alpha$ is a pure imaginary parameter:
\be
\alpha + \alpha^*=0, \qquad \beta=0. \lab{albeta_diag} \ee
Without loss of generality one can put $\alpha = -\, {\rm i}$, then 
\[
a_n= -\, 1, \quad b_n = 2 \sin(n\,{\theta}_N), \quad c_n = 1, 
\quad {\theta}_N: = \frac{2 \pi}{N}\,.
\]
Hence the matrix $X$ is diagonal,
\be
X\,= \,2\, diag\,(0,s_1,s_2,\cdots,s_{N - 2},s_{N - 1})\,,\qquad 
s_n = \frac{q^n-q^{- n}}{2{\rm i}} = \sin\left(n\,{\theta}_N\right)\,, \lab{s_n} 
\ee
whereas the {\lq}momentum{\rq} matrix $Y$ is tridiagonal (with zero entries on the main 
diagonal),
\be
Y= {\rm i}\left[\begin{array}{ccccc}
 0   &  - 1 & 0 & \cdots & 1 \\
 1   &  0 & - 1 & \cdots &  0  \\
\cdots & \ddots & \ddots & \ddots & \cdots \\
 0 & \cdots & 1 & 0 & - 1 \\
- 1 & \cdots & 0 & 1 & 0
\end{array}\right].
\lab{Y_exp} \ee
The final form of the intertwining operators 
\be
A= X + {\rm i} Y= \left[\begin{array}{ccccc}
  0  &   1     & 0 & \cdots & - 1 \\
- 1 & 2 s_1 & 1 &\cdots  &  0  \\
\cdots & \ddots & \ddots & \ddots & \cdots\\
  0 & \cdots & - 1 &2  s_{N-2} & 1 \\
 1 & \cdots & 0 & - 1 &2  s_{N-1}
\end{array}\right] 
\lab{A_fin} \ee
and
\be
A^{\dagger}\equiv A^{\intercal}= X - {\rm i} Y
= \left[\begin{array}{ccccc}
  0  & - 1 & 0 & \cdots & 1 \\
  1 & 2 s_1 & - 1 &\cdots  & \vdots  \\
\cdots & \ddots & \ddots & \ddots &\cdots \\
 0  & \cdots & 1 &2  s_{N-2} & - 1 \\
 - 1 & \cdots & 0 &  1 &2  s_{N-1}
\end{array}\right] 
\lab{AP_fin} \ee
thus coincide with those found in \cite{MesNat}. It remains only to add that the singular
(noninvertible) matrices $A$ and $A^{\intercal}$ are traceless, because one checks
easily that $\sum_{k=1}^{N - 1}s_k = 0$. It is a remarkable fact that the rank of the
matrix $A$ (which is the same as the rank of the matrix $A^{\intercal}$) turns out to
be different for the odd and even dimensions $N$:

1°. In the case of odd N the rank of the matrix $A$ is equal to $N - 1$ and the null space of
the matrix $A$ is one-dimensional;

2°. In the case of even N the rank of the matrix $A$ is equal to $N - 2$ and the null space of
the matrix $A$ is two-dimensional \cite{4Open}.

This unforeseen distinction between the properties of odd and even dimensional DFT's  simply
indicates that the discrete reflection symmetry in the vector spaces $\mathbb{R}^N$, spanned 
by the even-dimensional DFT eigenvectors, is spontaneously broken  \cite{SprProc}.

Recall  that the standard coordinate and momentum operators $x$ and $p$ are known to 
satisfy the Heisenberg commutation relation $[x,p]=i\,$. Contrary to this case, the operators 
$X$ and $Y$ do not satisfy such a simple commutation relation. In the next section we show 
that these operators satisfy instead a {\lq}classical{\rq} algebra of the Askey-Wilson type.

\section{Cubic algebra, generated by the operators $A$ and $A^{\intercal}$}
\setcounter{equation}{0}
The operators $A$ and $A^{\intercal}$, defined as in (2.8) and (2.9), respectively, 
do satisfy a simple {\it cubic} algebra under the commutation relation. Indeed, let us 
choose the third operator $C$ as the commutator
\be
C=[A, A^{\intercal}] = AA^{\intercal} - A^{\intercal} A. \lab{C} 
\ee 
The operator $C$ is represented by a 3-diagonal cyclic matrix,
\[
C= 4\left[\begin{array}{ccccc}
0 & s_1 & 0 & \cdots & s_1 \\
s_1 & 0  & s_2-s_1 & \cdots & \cdots  \\
0 & s_2-s_1 & 0  & s_3 - s_2 & \cdots  \\
0 & \cdots & \cdots & \cdots & 0 \\
\cdots & \cdots& s_{N-2} - s_{N-3} & 0 & s_{N-1} - s_{N-2}  \\
s_1 & \cdots & 0 & s_{N-1} - s_{N-2} & 0
\end{array}\right],
\]
and it is not hard to verify directly that
\be
[C,A] = \beta_1 A A^{\intercal} A + \beta_2 A - \beta_1 {(A^{\intercal})}^3, \quad [A^{\intercal},C] 
= \beta_1  A^{\intercal} A A^{\intercal} + \beta_2 A^{\intercal} - \beta_1 {A}^3, \lab{AB_alg} 
\ee
where
\be
\beta_1= \frac{(1-q)^2}{1+q^2}, \qquad \beta_2 = -\, 4 \: \frac{(q-q^{-1})^2}{q+q^{-1}}. \lab{par_AB} 
\ee
To check whether the three operators $ A$, $A^{\intercal}$ and $C$ form a closed algebra,
one should verify their compatibility with the Jacobi identity
\be
[A, [A^{\intercal}, C]] +  [C, [A, A^{\intercal}]] +  [A^{\intercal}, [C, A]] = 0. \lab{JId} 
\ee
Note that in our particular case the second double commutator in the Jacobi identity \re{JId}
vanishes because of definition (3.1). Hence the Jacobi identity \re{JId} reduces to the sum 
of only two double commutators, the first and the third one. So one evaluates next that 
\be
[A, [A^{\intercal}, C]] = \beta_1\, \Big\{ (A A^{\intercal})^2 - (A^{\intercal} A)^2\Big\}
+ \beta_2\, C  \lab{JId1} 
\ee
and
\be
[A^{\intercal}, [C, A]] = \beta_1\, \Big\{ (A^{\intercal} A)^2 -  (A A^{\intercal})^2\Big\}
- \beta_2\, C.  \lab{JId3} 
\ee
Consequently 
\be
[A, [A^{\intercal}, C]] +   [A^{\intercal}, [C, A]] = 0, \lab{JId13} 
\ee
and the Jacobi identity \re{JId} holds.

This algebra possesses the Casimir operator
\be
Q_1 = C^2 + r_1 \: \{A^2, (A^{\intercal})^2 \} + r_2 \: \{A, A^{\intercal}\} 
- r_1 \left(A^4 + (A^{\intercal})^4  \right) \lab{Q1} 
\ee
with
\be
r_1 = \frac{(1-q)^2}{(1+q)^2}, \qquad  r_2= -\, 4\, \frac{\left( 1+{q}^{2} \right) 
 \left( q-{q}^{-1} \right) ^{2} } {\left(1+ q \right) ^{2} }. \lab{r_Q} 
\ee
When $N \to \infty$, the parameter $q$ goes to 1 and the coefficients $\beta_1$ and 
$\beta_2$  tend to zero. This means that the commutation relations (3.2) become trivial,
\be
[C,A]=[A^{\dagger},C]=0 \lab{triv_AC}, \ee
as happens in the case of the linear quantum harmonic oscillator with the Heisenberg commutation 
relation $[a, a^{\dagger}]=1$ for the lowering and raising operators $a$ and $ a^{\dagger}$ . For 
arbitrary natural $N$, however, we see that the algebra \re{AB_alg} is not classical  and  does not 
even belong to the HAW type. Nevertheless, this algebra is still very simple and can be exploited to 
derive further useful relations concerning discrete Fourier transform.

\section{Askey-Wilson algebra for the operators $X,Y$}
\setcounter{equation}{0}
It is a remarkable fact that the operators $X$ and $Y$ are {\lq}classical{\rq} operators 
with nice spectral properties. For $X$ this is obvious because the spectrum of $X$ is 
\be
x_n = -{\rm i}\,(q^n - q^{-\,n}) = 2\,s_n, \qquad n=0,1,\dots, N-1 \lab{sp_x}, 
\ee
which indicates that $x_n$ belongs to the class of the Askey-Wilson spectra of the type
\be
\lambda_n = C_1 q^n +C_2 q^{-n} + C_0\,. \lab{AW_sp} 
\ee
Observe also that the eigenvectors of the operator $X$ are represented by the 
euclidean $N$-column orthonormal vectors ${e}_k$  with the components
$({e}_k)_l = {\delta}_{kl}$, $k,l = 0,1,\dots,N-1$, that is,
\be
X\,{e}_k = x_k\,{e}_k\,.\lab{evX} 
\ee
The spectrum $y_n$ of the matrix $Y$  belongs to the same Askey-Wilson family too.
Moreover, it turns out  that the spectra of the matrices $Y$ and $X$ are {\it equal}. 
The reason for this similarity is that the operators $X$ and $Y$ are unitary equivalent. 
Indeed, from the intertwining relations \re{inter} it follows that  \cite{MesNat} 
\be
Y \,\Phi =  \Phi\, X, \qquad X\,\Phi = - \,\Phi\, Y\,.\lab{inter_XY} 
\ee
This means that
\be
Y =  \Phi\, X\, \Phi^{\dagger},\lab{Y_X} 
\ee
i.e., that the operators $X$ and $Y$ are unitary equivalent and hence isospectral. 

Note that the spectrum of $X$ is simple (i.e., nondegenerate) only if $N$ is odd, whereas 
for $N$ even the spectrum of $X$ is doubly degenerate. Evidently, the same statement is 
true for the operator $Y$. 

Let ${\epsilon}_k, k=0,1,\dots, N-1$, be a set of the eigenvectors of the operator $Y$, 
\be
Y\, {\epsilon}_n = x_n \,{\epsilon}_n, \qquad n=0,1,\dots, N-1.\lab{evY} 
\ee
From the relations \re{evX} and \re{Y_X} it follows at once that the expansion relation
\be
{\epsilon}_n = \Phi\, e_n = \sum_{k=0}^{N-1} {\Phi_{kn}\, e_k} 
= N^{-1/2} \Big(1, q^n, q^{2n},\dots, q^{(N-1)n}\Big)^{\intercal} 
\lab{d_Phi_e}
\ee
between the sets of the eigenvectors $e_k$ and ${\epsilon}_n$ is valid. We thus see 
that the discrete Fourier matrix $\Phi_{kn}$ can be also defined as the matrix of the 
overlap coefficients between eigenbases of the operators $X$ and $Y$. Observe also
that it is not hard to show that the operator $Y$ is two-diagonal in the eigenbasis of
the operator $X$:
\be
Y\, {e}_n = {\rm i}\,(e_{n+1}\,-\,{e}_{n-1})\,,\lab{YebX} 
\ee
where $e_{-1} = e_{N-1}$ and $e_{N} = e_{0}$. Moreover, from \re{YebX} and 
\re{inter_XY} it follows that the operator $X$ is similarly two-diagonal in the eigenbasis 
of the operator $Y$:
\be
X\, {\epsilon}_n = {\rm i}\,({\epsilon}_{n-1}\,-\,{\epsilon}_{n+1})\,.\lab{XebY} 
\ee

Note that this symmetry between operators $X$ and $Y$ can be explained from the algebraic 
point of view in the following way.
\begin{pr}
The operators $X$ and $Y$ provide a representation of the Askey-Wilson algebra with 
commutation relations 
\ba
&&X^2Y +YX^2 -\left(q+q^{-1} \right) XYX = -\left(q-q^{-1}\right)^2 Y, \nonumber \\
&&Y^2 X + X Y^2 -\left(q+q^{-1} \right)YXY = -\left(q-q^{-1}\right)^2 X. \lab{AWF} 
\ea
\end{pr}
\begin{proof}
Taking into account that the componentwise structures of the operators $X$ and $Y$ in
the euclidean basis $e_k$  are of the form ({\it cf} \re{evX} and \re{YebX}, respectively)
$$
X_{kl} = x_k\,\delta_{k,l}, \qquad Y_{kl} = {\rm i}\,\Big(\delta_{k,l+1}\,-\,\delta_{k,l-1}\Big),
$$
one evaluates first that
\be
\Big(X^{2} Y + Y X^{2}\Big)_{kl}= \Big(x^2_k + x^2_l\Big)\,Y_{kl}, \qquad
\Big(XYX\Big)_{kl} = x_k x_l Y_{kl}.  \nonumber
\ee
Then 
\ba
\Big(X^{2} Y + Y X^{2} - (q - q^{-1}) XYX\Big)_{kl}= 
\Big(x^2_k + x^2_l - 2c_1 x_k x_l\Big)\,Y_{kl} 
= {\rm i}\Big( a_k \,\delta_{k,l+1}\,- \, a_{k+1}\delta_{k,l-1}\Big), \lab{a_k}
\ea
where $a_k= x^2_k + x^2_{k-1} - 2c_1 x_k x_{k-1}$. The last step is to show, by 
using the trigonometric identity $2c_k c_{2k-1}=c_{2k} + c_{2k-2}$, $c_k:=cos k{\theta}_N$,
that $a_k$ does not actually depend on the index $k$ since $a_k= 4s_1^2= - \,(q-q^{-1})^2$.
This means that the right-hand side in \re{a_k} is equal to $-\,(q-q^{-1})^2\,Y_{kl}$
and the first identity in \re{AWF} is proved.

Similarly, 
$$
\Big(Y^{2} X +  X Y^{2}\Big)_{kl}= \Big(x_k + x_l\Big)\,Y_{kl}^2
=(x_k + x_l)\Big(2\delta_{k,l}\,- \,\delta_{k,l+2}\,-\,\delta_{k,l-2}\Big),
$$
whereas
$$
\Big(YXY\Big)_{kl} = x_j Y_{kj} Y_{jl}=(x_{k+1} + x_{k-1})\,\delta_{k,l}
\,- \,x_{k-1}\delta_{k,l+2}\,-\,x_{k+1}\delta_{k,l-2}\,.
$$
Hence one concludes, upon using another trigonometric identity 
$x_{k+1} + x_{k-1}=4c_1s_k$, that
$$
\Big(Y^{2} X +  X Y^{2} - 2c_1 YXY\Big)_{kl}\,= 8 s_1^2 s_k\,\delta_{k,l}
\,=\,4 s_1^2 X_{k,l}\,=\,-\,(q-q^{-1})^2 X_{k,l}.
$$
Thus the second identity in \re{AWF} is proved as well.
\end{proof}
{\it Remark}. The generic Askey-Wilson algebra was introduced in \cite{Zhe_AW}. 
The above form of the Askey-Wilson algebra is due to Terwilliger \cite{Ter_AW}. 
Let us draw attention to the  remarkable symmetry property of relations  \re{AWF}: each 
can be obtained from another by the transposition $X \rightleftarrows Y$. 

The Askey-Wilson algebra can be presented in an equivalent  to \re{AWF} cyclic form 
if one introduces the third Hermitian operator $Z$ defined as
\be
Z= q^{1/2} \left[\begin{array}{ccccccc}
0 & q^{N-1} & 0 & \cdots & 0 & 0 & q^{N-1} \\
1 & 0 & q^{N-2} &  \cdots &  0 & 0 & 0 \\
0 & q & 0 &  \ddots  & 0 & 0 & 0 \\
\cdots & \cdots & \ddots & \ddots & \ddots &  \cdots &  \cdots \\
 0 & 0 & 0 & \ddots & 0 &  q^{2} & 0  \\
 0 & 0 & 0 & \cdots & q^{N-3} & 0 & q \\
 1 & 0 & 0 & \cdots &  0  &  q^{N-2} & 0 
\end{array}\right].\lab{Z}
\ee
Then the three Hermitean operators $X,Y,Z$ satisfy the Askey-Wilson algebra 
in the cyclic $\mathbb{Z}_3$ form
$$
p XY - p^{-1}YX = (q-q^{-1})Z, \qquad p ZX - p^{-1}XZ 
= (q-q^{-1})Y, $$ 
\be 
p YZ - p^{-1}ZY = (q - q^{-1})X, \lab{AW3} 
\ee
where $p=q^{1/2} = \exp\left({\pi}{\rm i}/N\right)$.This algebra can be 
considered as a q-analog of the rotation algebra $o(3)$. For $q$ a root of 
unity this algebra and its representations were considered in \cite{ZS_roots}.

Similar to the operator $Y$, the operator $Z$ is two-diagonal in the eigenbasis 
of the operator $X$:
\be
Z\, {e}_n = q^{1/2}\,\Big( q^n e_{n+1}\,+\,q^{-\,n} {e}_{n-1}\Big)\,.\lab{ZebX} 
\ee
Despite the completely symmetric form of the algebra \re{AW3}, the spectral 
properties of the opera\-tors $X,Y,Z$ are slightly different. 
\begin{pr}
By a similarity transformation $ S^{-1} Z S $ the matrix $Z$ can be transformed into 
the simple circulant matrix $(-1)^N\t Z$,
\be
\t Z=\left[
\begin{array}{cccccc}
 0 & 1 & 0 & \cdots & 0 & 1 \\
 1 & 0 & 1 &  \cdots & 0 & 0 \\
 0 & 1 & 0 & \ddots  & 0 & 0 \\
\cdots & \cdots & \ddots & \ddots &\ddots & \cdots  \\
 0 & 0 & 0 & \ddots & 0 & 1 \\
 1 & 0 & 0 & \cdots  & 1 & 0
\end{array}
\right],\lab{tZ} 
\ee
provided that the diagonal matrices $S$ and $S^{-1}$ are chosen componentwise 
as $S_{kl}= x_k\, \delta_{kl}$ and $S^{-1}_{kl}= x_k^{-1}\,\delta_{kl}$, with 
$x_k=(-1)^{kN} q^{k^2/2},\,\,0\leq k,l \leq N - 1 $.
\end{pr}
\begin{proof}
Taking into account that the componentwise structure of the matrix $Z$ in the euclidean basis 
$e_k$ is of the form $Z_{kl} = q^{-1/2}\,\Big( q^k \delta_{k,l+1}\,+\,q^{-\,k} \delta_{k,l-1}\Big)$
({\it cf} \re{ZebX}), one evaluates that 
\ba
\Big(S^{-1} Z S\Big)_{kl}= S^{-1}_{kn}Z_{nm} S_{ml} &=& x^{-1}_{k}x_l Z_{kl} \nonumber\\
=q^{-1/2}x^{-1}_{k}\,\Big( q^k x_{k-1}\delta_{k,l+1}\,+\,q^{-\,k} x_{k+1}\delta_{k,l-1}\Big)
&= &b_k\,\delta_{k,l+1}\,+\,b_{k+1}^{-1}\,\delta_{k,l-1}\,,\lab{1pr}
\ea
where $b_k = q^{k - 1/2}x^{-1}_{k}\,x_{k-1}$. Let us assume now that, in complete analogy
with the Proposition 1 ({\it cf} \re{a_k}), the coefficient $b_k$ does not depend on the index
$k$ and it is equal to $(-1)^N$:
\be
 b_k = q^{k - 1/2}x^{-1}_{k}\,x_{k-1} = (-1)^N\,.\lab{b_k}
\ee
Then from  \re{b_k} it follows at once that $x_k$ in this case have to satisfy a recurrence 
relation $x_{k+1}=(-1)^N q^{k+1/2} x_k$. Hence $x_{k}=(-1)^{kN} q^{k^2/2}$, 
$k=0,1,2,...,N-1$, and formula \re{1pr} reduces to
\be
\Big(S^{-1} Z S\Big)_{kl} = (-1)^N\Big(\delta_{k,l+1}
\,+\,\delta_{k,l-1}\Big)=(-1)^N \t Z_{kl} \,.\lab{2pr}
\ee
This completes the proof of the similarity between matrices $Z$ and $(-1)^N\t Z$.
\end{proof}
The spectrum $\zeta_n$ of the matrix $\t Z$ is well known \cite{Davis},
$$
\zeta_n = 2 \cos n\,{\theta}_N = 2 c_n, \quad n=0,1,2,...,N-1.
$$
So it is evident that apart from the eigenvalues  $\pm 2$, all other eigenvalues $\zeta_k$ 
of the matrix $Z$ are doubly degenerate because $\zeta_k = \zeta_{N-k}$, $k=1,2,...,N-1$, 
by definition. 

The Casimir operator of this algebra \re{AW3}, which commutes with $X,Y,Z$, 
has the expression
\be
Q= p\, XYZ - q \left( X^2 + Z^2 \right) - q^{-1} Y^2.
\ee
It is directly verified that $Q$ is Hermitian matrix, i.e., $Q^{\dagger} = Q$, and that 
for the given representation this operator is proportional to the identity matrix,
\be
Q = - \,2 \left( q+q^{-1} \right) \mathcal{I} = - \,4 \cos {\theta}_N \: \mathcal{I}.
\ee
 We thus have associated the {\lq}classical{\rq} Askey-Wilson algebra to the 
discrete Fourier transform.   

We close this section with the following extended remark about the algebra,
defined by \re{AW3}. Recall first that the Askey--Wilson polynomials reveal
themselves within the overlap functions between the two dual bases in the 
Askey--Wilson algebra \cite{Zhe_AW}. On the other hand, the algebra \re{AW3}
can be considered as a q-analog of the rotation algebra $o(3)$. Indeed, if one 
introduces three operators $K_0=\frac{1}{2s_1}X$,  $K_1=\frac{1}{2s_1}Y$ 
and $K_2=\frac{\rm i}{2s_1}Z$, then the commutation relations in \re{AW3}
can be rewritten as
\ba
[K_0,K_1]_q = K_2, \quad [K_0,K_2]_q = - K_1, 
\quad [K_1,K_2]_q = - K_0,       \lab{CRo3} \ea
where $[A,B]_q := q^{1/2} AB - q^{-1/2}BA$ by definition. The commutation
relations \re{CRo3} are known to define the $so_3(q)$ associative algebra for
$q$ a root of unity. This algebra and its representations were considered in 
\cite{ZS_roots}, where it was found that the above-mentioned overlap functions 
between the two dual bases of the operators $K_0$ and $K_1$ are expressed in 
this case in terms of the $q$-ultraspherical polynomials.

So the question can be then posed what is explicit form of a polynomial family, 
associated with the overlap functions between the two dual eigenbases of the 
operators  $K_0=\frac{1}{2s_1}X$ and $K_1=\frac{1}{2s_1}Y$ in our case. 
Now observe that from the relations \re{evX} and \re{d_Phi_e} it is evident 
that those overlap functions are equal to
\be
({\epsilon}_k, e_l)= N^{-1/2} \sum_{j=0}^{N-1} q^{- k j}  {\delta}_{l j}
= N^{-1/2} q^{- k l}=({\epsilon}_k, e_0) P_l({\mu}_k). \lab{ofXY}
\ee
Thus in the case under study the overlap functions between the two dual eigenbases of the 
operators $K_0$ and $K_1$ turn out to be the {\it monomials} $ P_l({\mu}_k) = {\mu}_k^l$ 
in the discrete argument $ {\mu}_k= q^{-\,k}$. These monomials $P_l({\mu}_k)$ form a 
complete and orthogonal system since
\be
N^{-1/2} \sum_{j=0}^{N-1}\,{P_k({\mu}_j)}^* \,P_l({\mu}_j) = {\delta}_{k l } . \lab{ocP}
\ee 

\section{Askey-Wilson-Heun operators and algebras}
\setcounter{equation}{0}
Starting from the pair of operators $X,$ which satisfy AW-algebra, one can construct the 
{\it algebraic Heun operator} $W$ as an arbitrary bilinear combination \cite{GVZ_H}
\be
W = \tau_1 XY + \tau_2 YX + \tau_3 X + \tau_4 Y + \tau_0 \mathcal{I} \lab{W} \ee
In \cite{Bas_H} it was shown that the pairs of operators $X,W$ or $Y,W$ constitute a cubic 
algebra closed under commutation relations. This algebra is the {\it Askey-Wilson-Heun} 
(briefly AWH) algebra.

We can apply this observation to our situation. Indeed, it is convenient to choose $X$ as 
one of the AWH generators because $X$ is the diagonal matrix with {\lq}classical{\rq} 
AW-spectrum. We can also choose $W$ as an operator {\it commuting} with $\Phi$. It is 
easily seen that the only possibility to get the operator $W$ commuting with $\Phi$ is to 
put $\tau_3=\tau_4=0$ and $\tau_2=-\tau_1$. This is equivalent to the choice of taking 
the commutator
\[
W = -2 i \: [X,Y] =  [A, A^{\dagger}].
\]
We show that the pair $X,W$ forms a simple cubic Heun-type algebra.  

Indeed, it is directly verified that the following relations hold
\be
X^2W+WX^2 -(q+q^{-1})XWX = g_1 W, \quad W^2X + XW^2 
- (q+q^{-1})WXW = g_2 X^3 + g_3 X, \lab{Heun_XW} 
\ee
where
\be
g_1=16s_1^2, \; g_2 = -16 c_1 (1+c_1)(1-c_1)^2, 
\; g_3= 64(1+c_1)(1-c_1)^2(3c_1+1). \lab{g_Heun} 
\ee
These relations correspond to a special case of the Askey-Wilson-Heun (AWH) 
algebra introduced in \cite{Bas_H}. 

From results of \cite{Bas_H} it is possible to derive the expression of the Casimir 
operator $Q$ which commutes with both $X$ and $Y$:
\be
Q= [X,W]^2 + \rho_1 \left((XW)^2 +(WX)^2  \right) 
+ \rho_2 W^2 + \rho_3 X^4 + \rho_4 X^2, \lab{Q_exp} 
\ee
where
\ba
&& \rho_1 = -\frac{(1-q)^2}{1+q^2}= \frac{1}{c_1} -1, \qquad  
\rho_2 = 4\, (q-q^{-1})^2 =-16 s_1^2, \nonumber \\
&& \rho_3 = \left( q+q^{-1} \right)^4 \, \frac{(1+q)^2}{1+q^2}
= 16 s_1^4 \left( 1+ c_1^{-1}\right), \nonumber \\
&& \rho_4 = - 4\,{\frac { \left( 1+q \right) ^{2} \left( 5\,{q}^{4}+2\,{q}^{3}
+2\,{q}^{2}+2\,q+5 \right)  \left( q-1 \right) ^{4}}{ \left( {q}^{2}+1
 \right) {q}^{4}}} = 64 s_1^2 \left( 1-c_1^{-1}\right) \left(5 c_1^2 +c_1 -2\right).
 \lab{rho_Q}  \ea
It is readily verified that for the given representations \re{s_n} and \re{Y_exp} of the 
matrices $X$ and $Y$, the operator $Q$ becomes the identity matrix to within a constant:
\be
Q = -64 \left(q-q^{-1} \right)^4 = -1024 \: s_1^4. \lab{Q_val} \ee

As is shown in \cite{Bas_H}, one can deduce useful information about the shape of the 
operators $X,W$, starting only with the commutation relations \re{Heun_XW}. Namely, 
it is possible to show that there exists a basis $e_n$ where the operator $X$ is diagonal 
while the operator $W$ is tridiagonal. On the other hand, in the {\lq}dual{\rq} basis 
$\t e_n$ (for which the operator $W$ is diagonal: $W {\t e_n} = \mu_n {\t e_n}$) the 
matrix of the operator $X$ will have, in general, a nonlocal shape with all entries being 
nonzero.

\section{Concluding remarks}
\setcounter{equation}{0}
To summarize, we have demonstrated that the {\lq}position{\rq} and {\lq}momentum{\rq} 
operators $X$ and $Y$ of the discrete Fourier transform form a special case of the Askey-Wilson 
algebra $AW(3)$. On the other hand, the creation and annihilation operators $A$ and 
$A^{\dagger}$ generate more complicated and {\lq}non-classical{\rq} algebra with cubic 
members in commutation relations. Moreover, we have shown that the position operator $X$ 
together with the operator $W=[A, A^{\dagger}]$, which commutes with the discrete
Fourier transform, constitute an algebra of Heun type. 

These results clarify what is an algebraic distinction of the discrete Fourier transform 
from the continuous one. In the classical case the commutator of creation and annihilation 
operators is equal to a constant and this leads to the well known Heisenberg-Weyl algebra,
which generates exact solutions of the quantum harmonic oscillator in terms of the Hermite
polynomials $H_n(x)$, times the Gaussian factor $exp(-\,x^2/2)$. 

Contrary to the continuous case, the creation and annihilation operators for the discrete Fourier 
transform $A^{\dagger}$ and $A$ do not form neither Lie algebra nor any of {\lq}classical{\rq} 
nonlinear algebras of Askey-Wilson type. This makes the problem of discrete Fourier harmonic 
oscillator, governed by the standard Hamiltonian $H = A^{\dagger}A + AA^{\dagger}$, hardly 
exactly solvable. This phenomenon of non-solvability of $H$ was observed earlier by a number of 
authors. Here we propose a simple algebraic explanation of this phenomenon. It remains but to 
mention in this connection that there is still a possibility to construct a version of the discrete 
Fourier harmonic oscillator but in terms of the eigenvectors of the difference operator $A^{\dagger}A$,  
upon interpreting it as a discrete analog of the number operator $N = a^{\dagger}a$ for the harmonic 
oscillator. Then it is possible to use the same procedure of constructing the eigenvectors of 
$A^{\dagger}A$ in the form of the ladder-type hierarchy, as in the harmonic oscillator case in
quantum mechanics: $a\,\psi_0(x) = 0$ and  $\psi_n(x) = 1/\sqrt{n}\, a^{\dagger}\psi_{n-1}(x), 
n= 1,2,3,...$. An important aspect to observe is that the operator $A^{\dagger}A$ commutes, 
$[A^{\dagger}A, P_d]=0$, with the discrete reflection operator $P_d$ for an arbitrary dimension
$N$. Nevertheless, thus constructed family of the eigenvectors of $A^{\dagger}A$ turn out to
be $P_d$-symmetric only for odd dimensions $N=2L+1$.In the all even cases with $N=2L$ the
$P_d$-symmetry in the space of the eigenvectors of $A^{\dagger}A$ is spontaneously broken.
This essential distinction between odd and even dimensions has been shown to be consistent with
the old formula \cite{McCPar}-\cite{DicSte} for the multiplicities of the eigenvalues, associated
with the $N$-dimensional DFT (consult \cite{MesNat}--\cite{SprProc} for a detailed discussion of 
this approach).

Note that the operators $X$ and $Y$ under discussion are similar to those considered 
earlier by Gr\"unbaum \cite{Gr_H}. These operators were exploited  in \cite{Gr_H}  in 
order to solve the  minimization problem of the uncertainty product $|\Delta X| |\Delta Y|$. 
Using a general approach for noncommuting compact operators $X,Y$ (see, e.g., \cite{Ifantis}), 
one can show that such minimization is achieved for generalized {\lq}coherent{\rq} states, 
that is, for the eigenvectors of the operator $V=X + i \gamma Y$ with an arbitrary real nonzero 
parameter $\gamma$. In the classical case, when $X$ is the coordinate operator and $Y$ the 
momentum, the solution of this problem for an arbitrary value of $\gamma$ is explicit and simple; 
it actually corresponds to the so-called squeezed states \cite{Nieto}. However, in our case the 
operator $V$ belongs to the class of  Heun operators defined by \re{W}. This leads to the conclusion 
that an explicit expression for such eigenvectors is hardly possible.  

Finally, it may be interesting to note that some general properties of representations of the 
Askey-Wilson algebra for roots of unity were considered in \cite{Huang} and \cite{Huang2}. 
Also, after completing this work, V.Spiridonov has drawn our attention to the paper \cite{DS}, 
where the intertwining relations for generators of Sklyanin algebra were established with 
respect to the so-called elliptic Fourier transform. It would be interesting to interrelate the 
results of the present work with those obtained in \cite{DS}.        

\keywords{}
\bigskip
{\Large\bf Acknowledgments}

We are grateful to E.Koelink, V.Spiridonov, J-F. Van Diejen, A.Veselov and L.Vinet for 
illuminating discussions.  MA thanks the support of Proyecto CONACYT No. A1-S-8793 
{\lq}Problemas de frontera para ecuaciones integrales y diferenciales parciales de 
an\'alysis y f\'isica matem\'atica{\rq} (Universidad Aut\'onoma del Estado de Morelos), 
NA thanks Project No.IG-100119 awarded by the Direcci\'on General de Asuntos del 
Personal Acad\'emico, Universidad Nacional Aut\'onoma de M\'exico. The work of AZ 
is supported by the National Science Foundation of China (Grant No.11771015). 
AZ also gratefully acknowledges the hospitality of the CRM over an extended period 
and the award of a Simons CRM professorship.

\bigskip

\bb{99}

\bi{McCPar} J.H.McClellan and T.W.Parks, {\it Eigenvalue and eigenvector decomposition 
of the discrete Fourier transform}, IEEE Trans. Audio Electroac., {\bf AU-20}, 66--74, 1972.

\bi{AusTol} L.Auslander and R.Tolimieri, {\it Is computing with the finite Fourier transform 
pure or applied mathematics?}  Bull. Amer. Math. Soc., {\bf 1}, 847--897, 1979.

\bi{DicSte} B.W.Dickinson  and K.Steiglitz, {\it Eigenvectors and functions of the discrete 
Fourier transform}  IEEE Trans. Acoust. Speech, {\bf 30}, 25--31, 1982.

\bi{Mehta} M.L.Mehta, {\it Eigenvalues and eigenvectors of the finite Fourier transform}, 
J. Math. Phys., {\bf 28}, 781--785, 1987.

\bi{Matv} V.B.Matveev, {\it Intertwining relations between the Fourier transform and
discrete  Fourier transform, the related functional identities and beyond}, Inverse Prob., 
{\bf 17}, 633--657, 2001.

\bi{Ata} N.M.Atakishiyev, {\it On $q$-extensions of Mehta's eigenvectors of the finite
Fourier transform}, Int. J. Mod. Phys. A, {\bf 21}, 4993--5006, 2006.

\bi{Sylv} J.J.Sylvester, {\it Thoughts on inverse orthogonal matrices, simultaneous sign
successions, and tessellated pavements in two or more colours, with applications to Newton's
rule, ornamental tile-work, and the theory of numbers}, Philos. Mag., {\bf 34}, 461--475, 1867.

\bi{MesNat} M.K.Atakishiyeva and N.M.Atakishiyev, {\it On the raising and lowering difference 
operators for eigenvectors of the finite Fourier transform}, J. Phys: Conf. Ser., {\bf 597}, 
012012, 2015.

\bi{4Open} M.K.Atakishiyeva, N.M.Atakishiyev and J.Loreto-Hern{\'a}ndez, {\it More 
on algebraic properties of the discrete Fourier transform raising and lowering operators}, 
4 Open, {\bf 2}, 1--11, 2019.

\bi{SprProc} M.K.Atakishiyeva, N.M.Atakishiyev and J.Loreto-Hern{\'a}ndez, {\it On the 
discrete Fourier transform eigenvectors and spontaneous symmetry breaking}, Springer 
Proceedings in Mathematics\,\&\,Statistics, {\bf 333}, 549--569, 2020.

\bi{Zhe_AW} A.S.Zhedanov, {\it {\lq}{\lq}Hidden symmetry{\rq}{\rq} of Askey-Wilson 
polynomials}, Theoretical and Mathematical Physics {\bf 89}, 1146--1157, 1991.

\bi{Ter_AW} P.Terwilliger, {\it The Universal Askey-Wilson Algebra}, SIGMA {\bf 7}, 
069, 2011, arXiv:1104.2813.

\bi{ZS_roots} V.Spiridonov and A.Zhedanov, {\it q-Ultraspherical polynomials for 
q a root of unity}, Lett.Math.Phys. {\bf 37}, 173--180, 1996.

\bi{Davis} P.J.Davis, Circulant Matrices, Wiley, New York, 1970.

\bi{GVZ_H} F. A. Gr\"unbaum, L. Vinet, and A. Zhedanov, {\it Algebraic Heun Operator and 
Band-Time Limiting}, Communications in Mathematical Physics {\bf 364}, 1041--1068, 2018, 
arXiv: 1711.07862.

\bi{Bas_H} P. Baseilhac, S. Tsujimoto, L. Vinet, and A. Zhedanov, {\it The Heun-Askey-Wilson 
Algebra and the Heun Operator of Askey-Wilson Type}, Annales Henri Poincar\'e {\bf 20} , 
3091--3112, 2019, arXiv: 1811.11407.

\bi{Gr_H} F.A. Gr\"unbaum, {\it The Heisenberg inequality for the discrete Fourier transform}, 
Appl. Comput. Harmon. Anal. {\bf 15}, 163--167, 2003.

\bi{Ifantis} E.K.Ifantis, {\it Minimal Uncertainty States for Bounded Observables}, 
J.Math.Phys. {\bf 12}, 2512--2516, 1971.

\bi{Nieto} M.M.Nieto and D.R.Truax, {\it Squeezed states for general systems}, 
Phys. Rev. Lett. {\bf 71}, 28--43, 1993.

\bi{Huang} H.Huang, {\it Finite-dimensional irreducible modules of the universal 
Askey-Wilson algebra at roots of unity}, Journal of Algebra {\bf 569}, 12--29, 2021.

\bi{Huang2} H.Huang, {\it Center of the universal Askey-Wilson algebra at roots of unity}, 
Nucl. Phys. B {\bf 909}, 260--296, 2016.

\bi{DS} S.E.Derkachov and V.P.Spiridonov, {\it Yang-Baxter equation, parameter 
permutations, and the elliptic beta integral},  Russ. Math. Surv. {\bf 68}, 10--27,
2013. arXiv: 1205.3520v2.

\end{thebibliography}

\end{document}